\documentclass[9pt,shortpaper,twoside,web]{ieeecolor}
\usepackage{generic}

\usepackage{cite}
\usepackage{amsmath,amssymb,amsfonts}
\usepackage{algorithmic}
\usepackage{graphicx}
\usepackage{textcomp}
\usepackage{graphicx}
\usepackage{color}
\definecolor{darkgreen}{rgb}{0,0.4,0}
\usepackage{float}
\definecolor{darkgreen}{rgb}{0,0.6,0}
\usepackage{authblk}
\usepackage[colorlinks,linkcolor=blue,citecolor=darkgreen]{hyperref}
\usepackage[latin1]{inputenc}
\usepackage{amsmath}
\usepackage{graphicx}
\usepackage{amssymb}
\usepackage{verbatim}
\usepackage{epsfig}%%
\usepackage[labelfont=bf]{caption}
\usepackage{enumerate}
\usepackage{subfig}
\usepackage{epstopdf}
\usepackage{relsize,exscale}
\def\qed{\hfill$\Box \Box \Box$}

\usepackage{booktabs} %% tables with three lines

\newtheorem{definition}{Definition}
\newtheorem{assumption}{Assumption}
\newtheorem{lemma}{Lemma}
\newtheorem{remark}{Remark}

\newtheorem{proposition}{Proposition}

\def\begquo{\begin{quote}}
\def\endquo{\end{quote}}
\def\begequarr{\begin{eqnarray}}
\def\endequarr{\end{eqnarray}}
\def\begequarrs{\begin{eqnarray*}}
\def\endequarrs{\end{eqnarray*}}
\def\begarr{\begin{array}}
\def\endarr{\end{array}}
\def\begequ{\begin{equation}}
\def\endequ{\end{equation}}
\def\lab{\label}
\def\begdes{\begin{description}}
\def\enddes{\end{description}}
\def\begenu{\begin{enumerate}}
\def\begite{\begin{itemize}}
\def\endite{\end{itemize}}
\def\endenu{\end{enumerate}}

\def\lef[{\left[\begin{array}}
\def\rig]{\end{array}\right]}
\def\qed{\hfill$\Box \Box \Box$}
\def\begcen{\begin{center}}
\def\endcen{\end{center}}
\def\begrem{\begin{remark}\rm}
\def\endrem{\end{remark}}
\def\begdef{\begin{definition}}
\def\enddef{\end{definition}}
\def\begpro{\begin{proposition}}
\def\endpro{\end{proposition}}
\def\begfac{\begin{fact}}
\def\endfac{\end{fact}}
\def\begass{\begin{assumption}}
\def\endass{\end{assumption}}
\def\begsubequ{\begin{subequations}}
\def\endsubequ{\end{subequations}}

\def\begmat#1{\begin{bmatrix}#1\end{bmatrix}}
\def\begali#1{\begin{align}{#1}\end{align}}

\def\cali{{\cal I}}

\def\calh{{\cal H}}

\def\caly{{\cal Y}}

\def\cals{{\cal S}}

\def\call{{\cal L}}

\def\sfh{\mathsf {H}}
\def\hatthe{\hat{\theta}}

\def\L2e{{\cal L}_{2e}}

\def\rea{\mathbb{R}}

\def\intnum{\mathbb{N}}

\def\diag{\mbox{diag}}
\def\adj{\mbox{adj}}
\def\col{\mbox{col}}
\def\hal{{1 \over 2}}

\def\diag{\mbox{diag}}
\def\rank{\mbox{rank}\;}

\def\ARC{{\it Annual Reviews in Control}}

\def\IJACSP{{\it Int. J. on Adaptive Control and Signal Processing}}

\def\TAC{{\it IEEE Trans. Automatic Control}}

\def\SCL{{\it Systems \& Control Letters}}
\def\AUT{{\it Automatica}}

\def\SIAM{{\it SIAM J. Control and Optimization}}

\usepackage{color}

\usepackage[prependcaption,colorinlistoftodos]{todonotes}

\def\BibTeX{{\rm B\kern-.05em{\sc i\kern-.025em b}\kern-.08em
    T\kern-.1667em\lower.7ex\hbox{E}\kern-.125emX}}
\begin{document}
\title{Conditions for Convergence of Dynamic Regressor Extension and Mixing Parameter Estimators Using LTI Filters}
\author{Bowen Yi and Romeo Ortega, \IEEEmembership{Life Fellow, IEEE}
%\thanks{This paragraph of the first footnote will contain the date on  }
%
\thanks{B. Yi is with Australian Centre for Field Robotics \& Sydney Institute for Robotics and Intelligent Systems, The University of Sydney, Sydney, NSW 2006, Australia (email: \texttt{bowen.yi@sydney.edu.au})
 }
\thanks{R. Ortega is with Departamento Acad\'{e}mico de Sistemas Digitales, ITAM, Ciudad de M\'exico, M\'{e}xico and Department of Control Systems and Informatics, ITMO University, Saint Petersburg 197101, Russia (email: \texttt{romeo.ortega@itam.mx})}}

\maketitle

\begin{abstract}
In this note we study the conditions for convergence of the recently introduced dynamic regressor extension and mixing (DREM) parameter estimator when the extended regressor is generated using LTI filters. In particular, we are interested in relating these conditions with the ones required for convergence of the classical gradient (or least squares), namely the well-known {\em persistent excitation} (PE) requirement on the original regressor vector, $\phi(t) \in \rea^q$, with $q \in \intnum$ the number of unknown parameters. Moreover, we study the case when only {\em interval excitation} (IE) is available, under which DREM, concurrent and composite learning schemes ensure global convergence, being the convergence for DREM in {\em finite time}. Regarding PE we prove, under some mild technical assumptions, that if $\phi$ is PE then the scalar regressor of DREM, $\Delta_N \in \rea$, is also PE ensuring exponential convergence. Concerning IE we prove that if $\phi$ is IE then $\Delta_N$ is also IE. All these results are established in the {\em almost sure} sense, namely proving that the set of filter parameters for which the claims do not hold is of zero measure. The main technical tool used in our proof is inspired by a study of Luenberger observers for nonautonomous nonlinear systems recently reported in the literature.
\end{abstract}

\begin{IEEEkeywords}
parameter estimation, system identification, adaptive control
\end{IEEEkeywords}

%
%%%%%%%%
\section{Introduction}
\label{sec1}
%%%%%%%%
%
We consider in the paper the problem of online estimation of the unknown parameter vector $\theta \in \rea^q$ from the linear regression equation (LRE)
\begequ
\label{lre}
 y(t) = \phi^\top(t) \theta,
\endequ
where $\phi(t) \in \rea^q$ is bounded and, for simplicity, we assume $y(t) \in \rea$. Here, we restrict ourselves to on-line recursive algorithms, which are attractive due to their robustness and direct applicability to adaptive control.

It is well-known \cite{SASBODbook} that the classical gradient estimator
$$
\dot {\hat \theta}(t)=\gamma \phi(t)[y(t)-\phi^\top (t)\hat \theta(t)],\;\gamma>0,
$$
ensures global exponential stability (GES) of the zero equilibrium of the associated linear time-varying (LTV) error equation
\begequ
\lab{errequ}
\dot {\tilde \theta}(t)=-\gamma \phi(t)\phi^\top(t)\tilde \theta(t),
\endequ
with $ {\tilde \theta}(t):=\hat \theta(t)-\theta$, if and only if the regressor $\phi$ is $(T,\delta)$-PE, that is, it satisfies
\begequ
\label{pe}
\int_{t}^{t+T} \phi(s) \phi^\top(s) ds \ge \delta I_q,
\endequ
for some $T>0,\;\delta>0$ and {\em all} $t \ge 0$---where we underscore the uniformity in time requirement. Another property of the gradient estimator that follows directly from \eqref{errequ} is monotonicity of the norm of the parameter error, that is,
\begequ
\lab{monpro}
|\tilde \theta(t_b)|\leq |\tilde \theta(t_a)|,\quad \forall t_b \geq t_a \geq 0.
\endequ

It has recently been shown \cite{BARORT,EFIBARORT,PRA17} that global asymptotic stability (GAS)---but not exponential---of the error equation \eqref{errequ} is ensured under the strictly weaker condition of {\em generalized PE}. Namely,
$$
\int_{\tau_k}^{\tau_{k+1}} \phi(s) \phi^\top(s) ds \ge \delta_k I_q
$$
where $\{\delta_k\}_{k \in \intnum}$ is a sequence of positive numbers, and  $\{\tau_k\}_{k \in \intnum}$ is a strictly increasing sequence of positive times such that $\tau_k \to \infty$ as $k \to \infty$, together with a technical assumption of the relation between $\delta_k$ and the integral of $|\phi(t)|^2$.
It is widely accepted that both, the PE and the generalized PE conditions, are extremely restrictive, a situation that has motivated the development of new estimation algorithms that relax these assumptions. The interested reader is referred to \cite{ORTNIKGER} for a recent survey of this literature and \cite{BOFSLO} for a novel interesting algorithm. The main objective of this paper is to establish the connection with the classical PE requirement and the new condition for convergence of the recently introduced DREM estimator  \cite{ARAetaltac17}.

\emph{Notation}. We use $\mathbb{C}$ to represent the complex plane, and $\mathbb{C}_{>0}$ for the open right half-plane. For a complex-valued matrix $A \in \mathbb{C}^{n\times m}$, $A^{\mathsf {H}}$ denotes the Hermitian transpose. Given a real-valued symmetric matrix $P \in \rea^{n\times n}$, $\lambda_{\tt min}(P)$ is its smallest eigenvalue. We use ${\bf 1}_{n}$ to represent an $n$-dimensional vector of ones.

%
%%%%%%%%
\section{Extended Regressor Equations and DREM}
\label{sec2}
%%%%%%%%
%
In this section we give the background material for the development of the DREM estimator.
\subsection{Extended LRE}
\lab{subsec21}
%%%%%%%%%%
%

A key component of all the new modified estimators is the construction of an {\em extended LRE} (ELRE). This idea was first reported in \cite{LIO} within the context of system identification and later used in \cite{KRE} for adaptive observers and in \cite{KREJOO} for adaptive controller designs. In both cases, the ELRE is created applying stable, linear  filters to the LRE \eqref{lre}. Namely,  we introduce a {\em linear}, bounded-input-bounded-output (BIBO), single-input $\ell$-output operator $\calh: \call_\infty \to \call_\infty^\ell$ to define the ELRE
\begequ
\lab{extlre}
Y(t) = \Phi(t) \theta,
\endequ
where $Y(t)  := \calh [y](t) \in \rea^{\ell}$ and the extended regressor matrix is defined as
\begequ
\label{yphi}
	\Phi(t)  := [\calh[\phi_1](t)~|~\ldots~|\calh[\phi_q](t)] \in \rea^{\ell \times q}.
\endequ

Applying a gradient-descent estimation to the ELRE \eqref{extlre} yields
$$
\dot{\hat\theta}(t)=\gamma \Phi^\top(t) [Y(t) - \Phi(t) \hatthe(t)],\;\gamma>0,
$$
whose corresponding parameter estimation error equation is
\begin{equation}
\label{errmoddre}
\dot{\tilde\theta}(t)=-\gamma \Phi^\top(t) \Phi(t) \tilde{\theta}(t).
\end{equation}
Notice that, in contrast to \eqref{errequ}, the  matrix $ \Phi^\top(t) \Phi(t) \in \rea^{q \times q}$ is not necessarily of rank one. This is the central property that motivates the extension of the regressor. However, it is well known \cite[Subsection 6.5.3(a.iv)]{NARANNbook} and \cite[Proposition 2]{ORTNIKGER}, that the provable stability properties of the error equation \eqref{errmoddre} are the same as the ones of \eqref{errequ}, and still require the PE condition for GES.\footnote{The only provable advantage of the new estimator is that the convergence speed can be improved increasing $\gamma$. However, as discussed in  \cite[Remark 6]{ORTNIKGER}, the interest of increasing the gain in adaptive systems is highly questionable.}

Two different ways to generate the ELRE have been studied in the literature. In \cite{LIO}, they used $\ell=q$ and the BIBO operator $\calh$ is obtained with the stable, linear {\em time-invariant} (LTI) filters
\begequ
\label{lti}
\calh_i(p) = {\lambda_i \over p + \lambda_i}, \; i \in \bar q:=\{1,\ldots,q\},
\endequ
with the differential operator $p:={d\over dt}$, $\lambda_i > 0$ and $\lambda_i \neq \lambda_j$ for all $i \neq j$. We refer in the sequel to the ELRE as Lion's (L-ELRE). Its state space realization can be written as
\begin{equation}
\label{ss_elre}
\begin{aligned}
\dot \Phi & =  -\Lambda \Phi + \Lambda \begmat{\phi^\top\\ \vdots \\ \phi^\top},
\\
\dot Y & = -\Lambda Y + \Lambda \col(y,\ldots, y),
\end{aligned}
\end{equation}
with $\Lambda = \diag(\lambda_1, \ldots, \lambda_\ell)$.

In \cite{KRE,KREJOO} they also consider $\ell=q$ and the operator $\calh$ is {\em LTV} of the form\footnote{It is clear that $\calh_i$ is an LTV operator of the form $\dot x_i=A_i x_i+b_i(t)u$ with $A_i=-\alpha$ and $b_i(t)=\phi_i$.}
\begequ
\label{ltv}
\calh_i(p,t) = {1 \over p + \alpha}\phi_i(t), \; i \in \bar q,
\endequ
with $\alpha>0$. A state space realization of Kreisselmeier's ELRE---called K-ELRE---is given by
\begequ
\label{dot_Phi}
\begin{aligned}
\dot \Phi(t) &=-\alpha \Phi(t)+\phi(t) \phi^\top(t)\\
\dot Y(t) &=-\alpha Y(t)+\phi(t) y(t).
\end{aligned}
\endequ

In recent years, ELREs have been used to ease the PE requirement in novel estimator schemes---see the recent survey in \cite{ORTNIKGER} for more details and a complete list of references. Two examples of these estimators are the concurrent \cite{CHOetal} and composite learning \cite{PANYU}. Within these methodologies, a dynamic data stack is built to discretely record online historical data, and the convergence of parameter estimation is managed monitoring the excitation over an interval. That is, the PE condition \eqref{pe} is replaced by the strictly weaker assumption that the regressor is IE, whose definition is given as follows.
\begin{definition}\rm
A bounded signal $\phi \in \rea^q$ is $(t_0,t_c,\mu)$-IE if there exist $t_0 \geq 0$ and $t_c>0$ such that
\begequ
\label{ie}
\int_{t_0}^{t_0+t_c} \phi(s) \phi^\top(s) ds \ge \mu I_q
\endequ
for some $\mu>0$.
\end{definition}

The IE condition is called ``exciting over a finite time interval'' in \cite[Definition 3.1, pp. 108]{TAO}, where it was used to analyze some convergence properties of the estimation error in the gradient algorithm.
\subsection{DREM estimator}
\lab{subsec22}
%%%%%%%%%%
%
A new estimator that has attracted a lot of attention, and has proven to be very successful to solve many theoretical and practical open problems is DREM, first proposed in \cite{ARAetaltac17} and recently reviewed in \cite{ORTetaltac20}. The main idea of DREM is to generate, out of the ELRE \eqref{extlre}, $q$ {\em scalar} LREs. Towards this end, we also fix $\ell = q$ and then introduce the key {\em mixing} step of multiplying from the left \eqref{extlre} by the adjugate of the (square) matrix $\Phi(t)$,  denoted $\adj\{\Phi(t)\}$, to get
\begin{equation}
\label{Yi}
\caly_i(t)=\Delta(t) \theta_i,\;i \in \bar q,
\end{equation}
where $\caly(t) :=\adj\{\Phi(t)\}Y(t)$ and
\begali{
\Delta(t) &:=\det\{\Phi(t)\}.
\label{del}
}
The DREM design is completed with the $q$ {\em scalar} estimators
\begequ
\label{estimator}
\dot {\hat \theta}_i(t)=\gamma_i \Delta(t)[\caly_i(t)-\Delta(t)\hat \theta_i(t)],\;\gamma_i>0, \; i \in \bar q,
\endequ
with associated error equations
\begequ
\lab{errequdrem}
\dot {\tilde \theta}_i(t)=-\gamma_i \Delta^2(t)\tilde \theta_i(t), \; i \in \bar q,
\endequ
for which the following proposition can be easily proved \cite{ARAetaltac17,ORTetaltac20}.

\begin{proposition}\rm
\lab{pro1}
The systems \eqref{errequdrem} enjoy the following feature.
\begenu
\item[\textbf{F1}] The origin is GAS $\;\Longleftrightarrow \;\Delta \notin \call_2$.
\item[\textbf{F2}] The origin is GES $\;\Longleftrightarrow\; \Delta$ is PE.
\item[\textbf{F3}] For all $t_b \geq t_a \geq 0$ we have $|\tilde \theta_i(t_b)| \leq |\tilde \theta_i(t_a)|,\;i \in \bar q$.
\endenu
\end{proposition}

Moreover, a variation of DREM that converges in {\em finite time} under the weaker IE assumption \eqref{ie} has been recently reported in \cite{ORTetalaut20} and, as discussed in  \cite{ORTetaltac20}, it has proven instrumental to solve many practical and theoretical open problems.

Note that in DREM we have replaced the convergence condition on the vector $\phi$ being PE by a condition---either non-square integrability for GAS or PE for GES---on a scalar quantity $\Delta(t)$, which is the determinant of the extended regressor matrix \eqref{del}. A natural question that arises is the relationship between the excitation properties of the original regressor $\phi(t)$ and the new scalar regressor $\Delta(t)$. This question has been recently answered in \cite{ARAetaltac20} for the case of K-ELRE where the following results are proven.
\begin{proposition}\rm
Consider the extended regressor matrix \eqref{yphi} generated via  \eqref{ltv}, and its determinant \eqref{del}.
\begenu
\item[\textbf{C1}] $\phi$ is PE  $\;\Longleftrightarrow \;\Delta$ is PE.
\item[\textbf{C2}] $\phi$ is $(t_0,t_c,\mu)$-IE  $\;\Longrightarrow \;\Delta$ is also $(t_0,t_c,\mu)$-IE.
\endenu
\end{proposition}

The results above prove that, in a scenario with suitable excitation, DREM with K-ELRE has the same convergence properties as the standard gradient, with the following additional advantages:
\begenu[(i)]
\item GAS under the non-square integrability condition of the scalar regressor $\Delta$ that, as shown in \cite[Proposition 3]{ORTetaltac20} is {\em strictly weaker} than PE of the regressor $\phi$;
\item element-by-element monotonicity of the parameter errors, which is {\em strictly stronger} than \eqref{monpro};
\item ability to tune, via $\gamma_i$, the convergence rate of {\em each} parameter error, in an independent way;
\item possibility to ensure {\em finite convergence time} under IE \cite[Proposition 3]{ORTetalaut20} without the injection of high-gain.
\endenu

The {\em main objective} of this paper is to prove a similar result for DREM with L-ELRE.\footnote{In \cite{ARAetalijacsp18} this question was studied for the particular case of systems identification, when the regressor is generated via LTI filtering of a sum of sinusoidal signals.}  In this way we conclusively establish the superiority of DREM---in either one of its forms, K- or L-ELRE---over classical estimators. Instrumental to establish our results is to adopt a  Kazantzis-Kravaris-Luenberger (KKL) observer perspective of the DREM estimator, as done in \cite{ORTetalaut18}. In this way, we can invoke a fundamental result on injectivity of the key mapping of KKL observers for {\em nonautonomous} nonlinear systems recently reported in \cite{BERAND}. This result extends to the nonautonomous case the previous results of \cite{ANDPRA} for autonomous systems, allowing then to include the study at hand.
%
%%%%%%%%%%%%%%%%%%
\section{Two Minor Modifications to the DREM Estimator}
\lab{sec3}
%%%%%%%%
%
To establish our results we introduce two slight modifications to the procedure described above. First, we do not select the number of filters {\em equal} to the dimension of the parameter vector, instead we set\footnote{This choice is made for simplicity, the results being true for any $\ell > q$.}
\begequ
\lab{ellq}
\ell = q+1.
\endequ
Moreover, in order to use the result of \cite{COR} in the proof of our main claim, we allow $\lambda_i$ to be in $\mathbb{C}_{>0}$. This modification is similar to the procedure used in the design of KKL observers, first proposed in \cite{KAZKRA} and intensively studied in \cite{ANDPRA,BERbook} where---to ensure injectivity of a key mapping---the number of LTI filters is selected {\em larger} than the dimension of the systems state. Such an approach was suggested in \cite[Section 5]{ORTetalaut18}, where the DREM estimator is revisited as a KKL observer for the LTV system
\begali{
\nonumber
\dot{\theta}& = 0\\
\label{ltvkkl}
y(t) &= \phi^\top(t) \theta,
}
in which the ``state" $\theta$ is constant and the ``output matrix" is $\phi^\top (t)$. As pointed out in \cite{ORTetalaut18} the results on KKL observers for {\em nonautonomous} systems,  known at that time, were insufficient to carry out the analysis of the DREM estimators. Fortunately, this situation has evolved, and the issue has been fully addressed in \cite{BERAND}. In particular, we will show in the paper that \cite[Theorem 3]{BERAND} can be easily adapted to answer our questions. For ease of reference, a simplified version of this fundamental result is given in the Appendix.

Our second modification is a consequence of the choice of \eqref{ellq}. Indeed, in order to apply the mixing step of DREM described in Subsection \ref{subsec22}, it is obviously necessary to have a {\em square} extended regressor matrix. This is easily achieved premultiplying the ELRE \eqref{extlre} by $\Phi^\sfh(t)$ and {\em redefining} the LREs \eqref{Yi} as
$$
\caly_N(t)=\Delta_N(t) \theta,
$$
in which $\caly_N(t) := \adj\{\Phi^\sfh(t) \Phi(t)\} \Phi^\sfh(t) Y(t)$ and
\begequ
\label{deln}
\Delta_N(t) := \det\{\Phi^\sfh(t) \Phi(t)\} \in \rea,
\endequ
where the subscript $(\cdot)_N$ is added to underscore that these are new signals. It is clear that we can still follow the last step in DREM to estimate the parameters $\theta_i$ individually. One important observation is that, since $\Phi(t)$ is a {\em tall} matrix the following implication is true
\begequ
\lab{ranphi}
\rank\{\Phi(t)\}=q 
\quad \Longrightarrow
\quad
\rank\{\Phi^\sfh(t)\Phi (t)\}=q.
\endequ
This implication allows us to study the properties of the ``new" $\Delta_N(t)$ defined in \eqref{deln} via the analysis of the rank of $\Phi(t)$.

A final observation is that, without loss of generality for the purposes of this note, we assume that all filters initial conditions are zero.

\begin{remark}
The two modifications proposed in this section are minor from the constructive perspective. However, we consider the case $\ell > q$, since it was shown in \cite{BERAND,ANDPRA} that the excessive coordinates play an important role to guarantee the key mapping, discussed in the following section, being injective. In contrast, there is no provable guarantee for the case $\ell =q$.
\end{remark}
%
%%%%%%%%%%%%%%%%%%%%%%
\section{Main Result}
\label{sec4}
%%%%%%%%%%%%%%%%%%%%
%
We are now in position to present the main result of this note.

\begin{proposition}
\label{pro3}\rm
Consider the extended regressor matrix \eqref{yphi} generated via \eqref{lti}, with $\lambda_i \in \mathbb{C}_{>0}$, $\lambda_i \neq \lambda_j$ for $i \neq j$ and $\ell = q+1$, and the determinant \eqref{deln}. For almost all choices of $\lambda_i$ the following implications hold true.\footnote{The qualifier ``almost" stems from the fact that the set of $\lambda_i$ for which the implications do not hold has {\em zero Lebesgue measure}.}
\begin{itemize}
    \item[\textbf{P1}] If $\phi$ is $(t_0,t_c,\mu)$-IE, then there exists a moment $t_\star>0$ such that
    \begequ
    \lab{delp1}
    \Delta_N(t) > 0,\quad \forall t\in [t_\star, \infty).
    \endequ
    %
  %  \vspace{0.1cm}

    \item[\textbf{P2}] If $\phi$ is $(T,\delta)$-PE and smooth with its time derivatives bounded, and if for any interval $[t,t+T)$, $\forall t\ge 0$ there exists a sub-interval $[t_c, t_c + \delta_c] \subset [t, t+T)$ such that 
    \begequ
    \label{Lh}
    \begin{aligned}
    \lambda_{\tt min }\left( \phi^\top_H(s) \phi_H(s) \right) \ge  L_H^{-1}, \quad \forall s \in [t_c, t_c + \delta_c]
    \end{aligned}
    \endequ
    with 
    $
    \phi_H:= \begmat{\phi(s)& \dot \phi(s) & \ldots & \phi^{(q-1)}(s)},
    $
    for some $\delta_c>0$ and some $L_H>0$. Then, $\Delta_N$ is PE.
\end{itemize}
\end{proposition}
\begin{proof}
The proof of {\bf P1} is based on \cite[Theorem 3]{BERAND}, which relies on the deep and fundamental result on almost sure injectivity of mappings reported in \cite[Lemma 3.2]{COR}. We present in the Appendix a simplified version of the result of \cite{BERAND} that is suitable for our analysis.

The gist of the proof is to verify, for our DREM scenario, the conditions of Lemma \ref{lem1}, and it proceeds along the following basic steps. First, we identify the mapping \eqref{T} of Lemma \ref{lem1} for the LTV system \eqref{ltvkkl}. Second, we prove that, if the extended regressor matrix $\Phi(t)$ is full-rank, the mapping is injective selecting the filter parameters outside a zero-Lebesgue measure set. Third, we prove that the assumption of $\phi$ being IE guarantees the backward-distinguishability condition in Lemma \ref{lem1}. The proof is completed invoking the implication \eqref{ranphi}. For convenience and with a slight abuse of notations, the output of the system \eqref{ltvkkl}, parameterized by $\theta$, is rewritten as $y_\theta(t)$.

Some simple calculations show that, for the system \eqref{ltvkkl}, the mapping \eqref{T} takes the form
$$
\begin{aligned}
T(\theta,t) & =  \mathop{\mathlarger{\mathlarger{\mathlarger{\int_{0}^{t}}}}}  \begmat{e^{\lambda_1(s-t)} y_\theta(s) \\ \vdots \\ e^{\lambda_\ell(s-t) }y_\theta(s)} ds\\
& =  \mathop{\mathlarger{\mathlarger{\mathlarger{\int_{0}^{t}}}}}  \begmat{e^{\lambda_1(s-t)} \phi^\top(s) \\ \vdots \\ e^{\lambda_\ell(s-t) }\phi^\top(s)} \theta ds\\
& = \Lambda^{-1}
\begmat{ \displaystyle \int_0^t \lambda_1 e^{\lambda_1(s-t)} \phi^\top(s) ds\\ \vdots \\ \displaystyle \int_0^t    \lambda_\ell e^{\lambda_\ell(s-t) }\phi^\top(s)ds}
\theta \\
& = \Lambda^{-1} \Phi(t) \theta,
\end{aligned}
$$
where $\Lambda:=\diag\{\lambda_i\} \in \mathbb{C}^{\ell \times \ell}$. Clearly, according to Definition \ref{def:inj}, for a given moment $t$, injectivity holds if and only if $\Phi(t)$ is full rank, {\em i.e.}
\begequ
\lab{injiffran}
T(\cdot,t)\;\mbox{ is\;injective}\quad \Longleftrightarrow \quad \rank\{\Phi(t)\}=q.
\endequ

We proceed now to prove that, under the condition of $\phi$ being $(t_0,t_c,\mu)$-IE, the backward-distinguishability condition of Lemma \ref{lem1} is satisfied. First, select $t_\star := t_0 + t_c$, then for $t\ge t_\star$, we have
\begin{equation}
\label{int_0t}
\begin{aligned}
\int_{0}^t \phi(s)\phi^\top(s) ds &  = \int_{0}^{t_0} \phi(s)\phi^\top(s) ds+ \int_{t_0}^{t_\star} \phi(s)\phi^\top(s) ds \\
& \quad +
\int_{t_\star}^{t} \phi(s)\phi^\top(s) ds \\
& \ge \int_{t_0}^{t_\star} \phi(s)\phi^\top(s) ds\\
& \ge \mu I_q.
\end{aligned}
\end{equation}
Consider two parameter vectors $(\theta_a,\theta_b) \in \rea^{2q}$. Then, we have
\begequ
\label{key_implication}
\begin{aligned}
& y_{\theta_a}(s) - y_{\theta_b}(s)  \equiv 0, \; \forall s \in [0,t] \quad \\
&\qquad \qquad \qquad  \Longleftrightarrow \quad
\phi^\top(s) (\theta_a -\theta_b) \equiv 0, \; \forall s \in[0,t]\\
& \qquad \qquad \qquad \Longrightarrow \quad
\phi(s)\phi^\top(s) (\theta_a -\theta_b) \equiv 0, \; \forall s \in[0,t]\\
& \qquad \qquad\qquad  \Longrightarrow \quad
\int_0^t\phi(s)\phi^\top(s) ds (\theta_a-\theta_b) =0\\
& \qquad \qquad \qquad \overset{\eqref{int_0t}}{\Longrightarrow} \quad \theta_a =\theta_b.
\end{aligned}
\endequ

We now show that the $(t_0,t_c,\mu)$-IE condition implies backward-distinguishability in $t_u$ (to be defined) during a compact interval. Without loss of generality, we assume $t_0>0$.\footnote{For the case $t_0 =0$, if the $\phi(t)$ is continuous, we can always find a sufficiently small $\Delta t>0$ such that the given signal is $(\Delta t, t_c,\mu)$-IE.} Define a parameter $t_u := \hal t_0 + t_c$, and we consider the compact interval 
$$
\cali_1:= \left\{t : t_0 + t_c \le t \le {3\over 2}t_0 + t_c \right\}.
$$
It is easy to show
$$
[t_0, t_0+ t_c] \subseteq [t- t_u, t] \quad \forall t \in \cali_1.
$$
Hence, following the same procedure in \eqref{key_implication}, we have for any $t\in \cali_1$
$$
\begin{aligned}
& y_{\theta_a}(s) - y_{\theta_b}(s)  \equiv 0, \; \forall s \in [t-t_u, t] \\
& \qquad \qquad \hspace{3cm} \qquad {\Longrightarrow} \quad \theta_a =\theta_b.
\end{aligned}
$$
The above implication shows that the LTV system \eqref{ltvkkl} is backward-distinguishable in $t_u$ for $t\in \cali_1$. Then invoking Lemma \ref{lem1}, we have that $T(\cdot,t)$ is injective in the bounded interval $\cali_1$. From the above, we have that the IE condition implies the injectivity at least in a small interval, \emph{i.e.},
\begin{equation}
\label{implication2}
\begin{aligned}
\phi \mbox{~is~} (t_0,t_c,\mu)\mbox{-IE} ~\implies~
\mbox{Injectivity of~} T(\cdot,t),\\t \in [t_0+t_c, {3\over 2}t_0 + t_c]
\end{aligned}
\end{equation}
for $t_0>0$. On the other hand, it is widely known that if $\phi$ is $(t_0,t_c,\mu)$-IE, then it is also $(t_0,t_c + T_c,\mu)$-IE for any $T_c\ge 0$.\footnote{If the IE condition is parameterized by different $T_c\ge 0$, the parameter $t_u$ in backward-distinguishability will be different as well.} Therefore, $T(\cdot,t)$ is injective for any finite-time moment $t\ge t_\star = t_0+t_c$ by considering different values of $T_c$. The proof of \eqref{delp1} is completed recalling the equivalence \eqref{injiffran} and the implication \eqref{ranphi}.

We now give the proof of {\bf P2}, which is mainly based on \cite[Theorem 2]{BERAND}. That result is about the injectivity of $T(\cdot,t)$ for the nonlinear time-varying system \eqref{nlsys} with {strong differential observability}---known as the foundation of high-gain observers---which can be used to impose uniform injectivity. Condition \eqref{Lh} guarantees that $\phi_H(t)$ is full rank in the intervals $[t_c, t_c+ \delta_c]$. As a result, the LTV system \eqref{ltvkkl} is strongly differentially observable in $[t_c, t_c+ \delta_c]$.

Since all the time-varying signals involved and their time derivatives are bounded, it is straightforward to verify Items 1)-2) and 4) in \cite[Assumption 3]{BERAND}, which is the key assumption used in \cite[Theorem 2]{BERAND}, for the system \eqref{ltvkkl}, in which the output function is given by 
$$
h(\theta,t)= \phi^\top(t)\theta.
$$
Together with \eqref{Lh}, Item 3) in \cite[Assumption 3]{BERAND} also holds in these intervals $[t_c,t_c+\delta_c]$.

Comparing to the sufficient conditions used in \cite[Theorem 2]{BERAND}, the only part that remains to be shown is the controllability of the pair $(\Lambda, B)$ with $B= {\bf 1}_{q+1}$. The associated controllability matrix is given by
$$
\begin{aligned}
\begmat{B & \Lambda B & \ldots & \Lambda^q B}
= 
\begmat{1 & \lambda_1 & \ldots & \lambda_1^q
\\
1 & \lambda_2 & \ldots & \lambda_2^q
\\
\vdots & \vdots& \ddots &\vdots
\\
1 & \lambda_{q+1} & \ldots & \lambda_{q+1}^q
},
\end{aligned}
$$
which is a Vandermonde matrix, thus being full rank due to the assumption on distinct $\lambda_i$. Hence, the pair $(\Lambda, B)$ is controllable, then verifying all the assumptions required in \cite[Theorem 2]{BERAND}.

According to that theorem there exists a constant $L_k>0$ such that the mapping $T(\cdot,t)$ defined in \eqref{T} is an injective immersion with
$$
|\theta_a - \theta_b| \le L_k |T(\theta_a,t) - T(\theta_b,t)|,\quad \forall t \in [t_c + t_k, t_c+\delta_c]
$$
for any $(\theta_a,\theta_b) \in \rea^{2q}$ and some $0<t_k < \delta_c$ by properly selecting $L_H>0$. (In the proof of \cite[Theorem 2]{BERAND}, it was shown that $t_k$ is decreasing with smaller $L_H$.)

Then, in the above interval, we have
$$
|\Lambda^{-1}\Phi(t)(\theta_a - \theta_b)| \ge L_k^{-1} |\theta_a - \theta_b|
$$
for any $(\theta_a,\theta_b) \in \rea^{2q}$. Defining $v:=\theta_a - \theta_b$, we get
\begequ
\label{vv}
v^\top \Phi^\sfh(t) \Phi(t) v \ge L_0  v^\top v
\endequ
for some $L_0>0$. Since \eqref{vv} holds for any $v \in \rea^q$, we have
$$
\Phi^{\sfh}(t) \Phi(t) \ge L_0  I_{n \times n}.
$$
Noting $\Delta_N = \det(\Phi^\sfh \Phi)$, we conclude that $\Delta_N$ has a uniformly positive lower bound in these intervals, {\em i.e.}, $\forall t\ge 0$
$$
\Delta_N(s) \ge L_1, ~ \forall s\in [t_c + t_k , t_c + \delta_c] \subset [t,t+T),
$$
for some $L_1>0$ independent of $t$. Invoking the uniformity with respect to time in \eqref{Lh}, and the definition of PE, we complete the proof.
\end{proof}

The following remarks are in order.

\begin{remark}
It is important to note that the claim {\bf P1} proves that if $\phi$ is IE then $\Delta_N(t)$ is bounded away from zero for all {\em finite} times $ t$, which {\em does not} ensure that the necessary and sufficient condition for convergence, {\em i.e.}, $\Delta_N \notin \call_2$, is satisfied. On the other hand, if $\phi$ is PE and satisfies the condition \eqref{Lh}, using successively {\bf P2} and {\bf F2} of Proposition 1, it ensures that the DREM estimator is GES. These facts are illustrated in simulations of the section below.
\end{remark}

\begin{remark}
If $\phi\in $ PE, it is possible to estimate the unknown parameter vector $\theta$ off-line from the L-ELRE or the scalar mixed regressors. Indeed, as done in KKL observer \cite{BERAND},  we can replace the online gradient-descent search by the calculation of the inverse mappings, that is,
$$
\hat \theta(t) = [\Phi^\sfh(t) \Phi(t)]^{-1}\Phi^\sfh(t)Y(t)
$$
for the former and the simpler $\hat \theta(t) = {1\over \Delta_N(t)}\caly_N(t)$ for the latter. However, such an implementation is fragile in the presence of noise and stymies the ability to track (slowly) time-varying parameters.
\end{remark}

%%%%%%%%%

\begin{remark}
The key point in {\bf P1} is to establish the connection between an interval property of $\phi$ and the ``point-wise'' property of the scalar variable $\Delta_N$ after some moment of time. To give some insight on this fact we recall that we can apply---as done in \cite[Proposition 7]{ORTNIKGER} for KRE---Jacobi's formula to \eqref{ss_elre} (with $\ell=q$ and $\Delta$ defined in \eqref{del}) to get
$$
\dot \Delta  =  -\lambda_\Sigma \Delta + {\bf v}_\lambda^\top \adj\{\Phi\}\phi
$$
with $\lambda_\Sigma = \sum_{i}^{\ell} \lambda_i$ and ${\bf v}_\lambda = \col(\lambda_1, \ldots, \lambda_\ell)$. Then, solving the equation above, we have
$$
\Delta(t) = \Delta(t_0)e^{-\lambda_\Sigma (t-t_0)} + \int_{t_0}^t e^{\lambda_\Sigma(s-t)} {\bf v}_\lambda^\top \adj\{\Phi(s)\}\phi(s) ds.
$$
Clearly, the ``point-wise'' function $\Delta(t)$ gathers the interval information of $\phi$ in the past. 
\end{remark}

\begin{remark}
The condition {\bf P1}, does not guarantee that the scalar regressor $\Delta_N$ is PE, only that it is IE. Hence, using the standard gradient descent we cannot ensure that the estimation error converges to zero. This problem has been recently overcome in  \cite{BOBetal}, where with some additional filtering operations we generate from LRE with IE regressors, new LRE with PE regressors.
\end{remark}

%
%%%%%%%%
\section{Simulation Results}
\label{sec4}
%%%%%%%%
%
To illustrate the results of Proposition \ref{pro3}, in this section we present some numerical simulations for the LRE \eqref{lre} with $q=2$ and $\theta = \col(-1,2)$. The DREM estimator with L-ELRE is designed selecting $\ell=3$, the filter parameters $\lambda_i$ equal to $0.2, ~0.3$ and $0.4$ and both adaptation gains set as $\gamma_i=2$. All initial conditions---of the filters and the estimated parameters---are set equal to zero.

First, we consider the PE case, selecting the regressor $\phi(t) = \col(5\sin(t), 8\cos(t))$. It is very easy to check that the given signal $\phi(t)$ satisfies the full rank condition \eqref{Lh}. Fig. \ref{fig:1} shows the evolution of $\Delta_N(t)$, which is positive for all $ t\in [t_\star,\infty)$, with $t_\star=2\pi$. The lower bound of $\Delta_N(t)$ is uniformly larger than zero even for $t\to\infty$, verifying the statement {\bf P2}. We also show in Fig. \ref{fig:2} the trajectories of the DREM estimation errors $\tilde{\theta}_i$, clearly revealing the element-by-element exponential convergence.

\begin{figure}[h]
    \centering
    \includegraphics[width=0.5\textwidth]{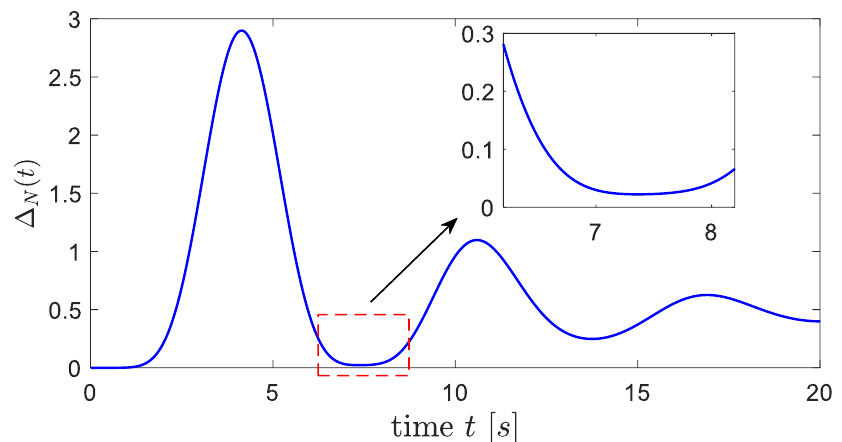}
    \caption{Evolution of the mixing regressor $\Delta_N(t)$ in the PE case}
    \label{fig:1}
\end{figure}

\begin{figure}[h]
    \centering
    \includegraphics[width=0.5\textwidth]{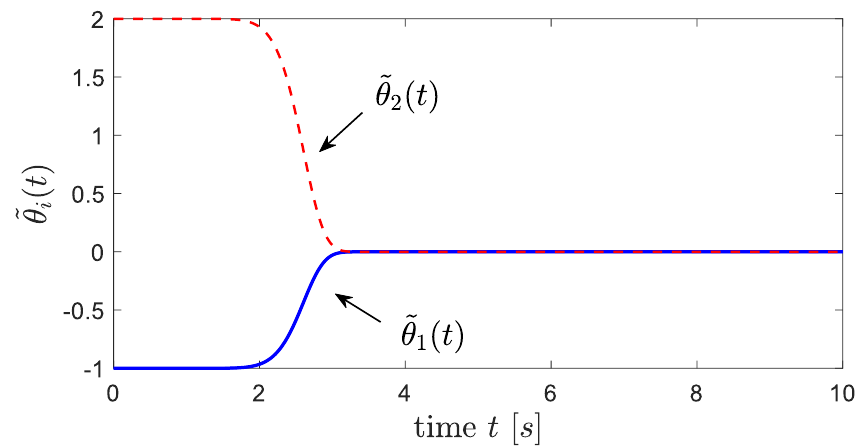}
    \caption{Trajectories of the estimation errors in the PE case with adaptation gains $\gamma_i=2$ ($i=1,2,3$)}
    \label{fig:2}
\end{figure}

Second, we consider the case when $\phi$ is $(t_0,t_c,\mu)$-IE considering the regressor
$$
\phi(t) = \left\{
\begin{aligned}
&\begmat{5\sin(t) \\ 8\cos(t)}, & \quad  t\in [0,5)\\
& \begmat{0 \\ 0},  & t \geq 5.
\end{aligned}
\right.
$$
The behavior of $\Delta_N(t)$ is shown in Fig. \ref{fig:3}. As predicted by the theory, there exists a time $t_\star$ such that $\Delta_N(t)>0$ for all $t \in [t_\star,\infty)$. On the other hand, since the excitation vanishes after $t=5$,  $\Delta_N(t)$ asymptotically converges to zero and thus the estimated parameters do not converge. Fig. \ref{fig:4} illustrates the trajectories of the DREM estimation errors $\tilde{\theta}_i$ with $\gamma_i=0.2$---as expected, taking larger $\gamma_i$ will reduce the steady-state error, see Fig. \ref{fig:5} obtained with $\gamma_i = 0.35$.

\begin{figure}[h]
    \centering
    \includegraphics[width=0.5\textwidth]{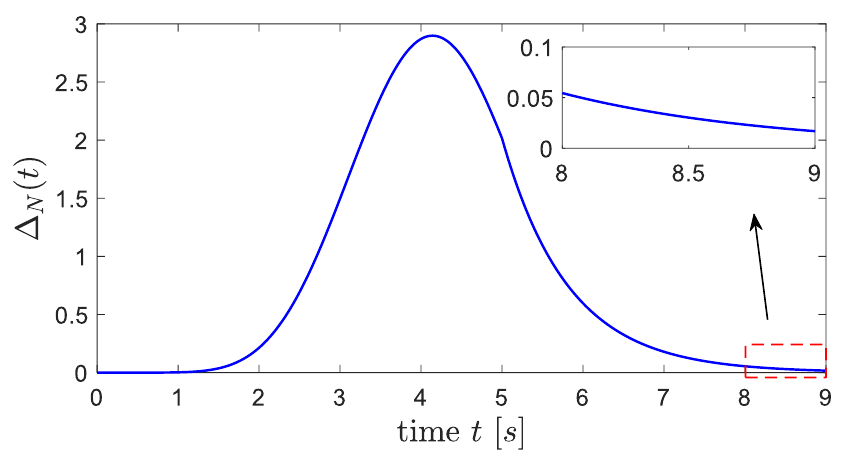}
    \caption{Evolution of the mixing regressor $\Delta_N(t)$ in the IE case}
    \label{fig:3}
\end{figure}

\begin{figure}[h]
    \centering
    \includegraphics[width=0.5\textwidth]{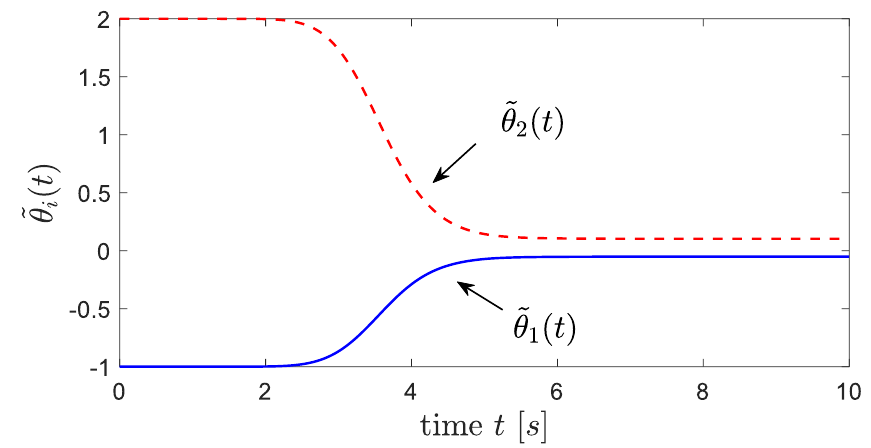}
    \caption{Trajectories of the estimation errors in the IE case with adaptation gains $\gamma_i=0.2$ ($i=1,2,3$)}
    \label{fig:4}
\end{figure}
\begin{figure}[h]
    \centering
    \includegraphics[width=0.5\textwidth]{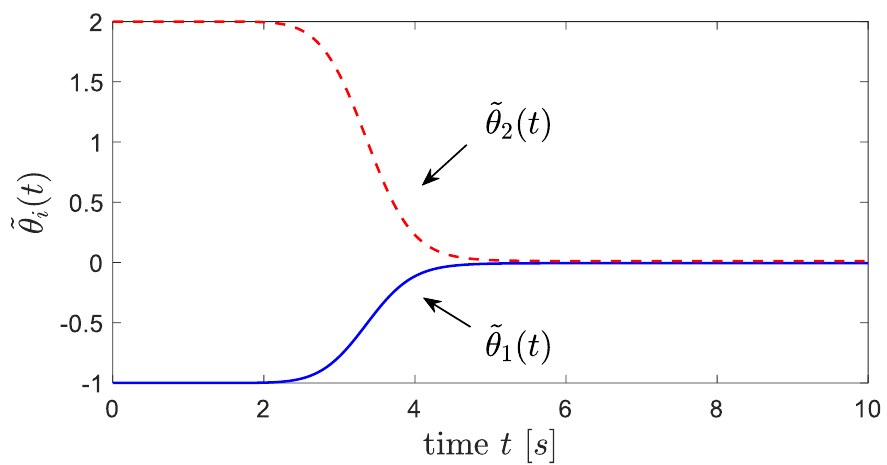}
    \caption{Trajectories of the estimation errors in the IE case with adaptation gains $\gamma_i=0.35$ ($i=1,2,3$)}
    \label{fig:5}
\end{figure}

\appendix
\begin{definition}\rm
\label{def:inj}
A (time-varying) function $T: \rea^n \times \rea_{\ge 0}  \to \rea^{m}$ is injective for fixed time $t$ if
$$
|x_a - x_b | \le \rho (|T(x_a,t) - T(x_b,t)|)
$$
for all $(x_a,x_b) \in \rea^n \times \rea^n$ and $t\ge 0$, where $\rho$ is a concave $\mathcal{K}^\infty$ function.
\end{definition}

\begin{lemma}\rm
\label{lem1}
\cite[Theorem 3]{BERAND} Consider the nonlinear time-varying system
\begequ
\lab{nlsys}
\dot{x}  = f(x,t), \quad y=h(x,t)
\endequ
with $x\in \rea^{n}$ and $y\in \rea$. Assume it is forward complete and backward-distinguishable in a bounded interval $\cali$ with $t\ge t_u$ and $t_u>0$, {\em i.e.}, for any $t\in \cali$ with $t\ge t_u$ and any $(x_a,x_b)\in \rea^{2n}$
$$
y_{x_a}(s) = y_{x_b}(s), \quad \forall s \in [t-t_u,t] ~\Longrightarrow~ x_a=x_b,
$$
where $y_{x_0}(t)$ denotes the output $y(t)$ for a state trajectory starting at $x(0)=x_0$. Then, there exists a {\em zero-Lebesgue measure} set $\cals \subset \mathbb{C}^{n+1}$ such that for any $(\lambda_1,\ldots, \lambda_{n+1}) \in \mathbb{C}_{>0}^{n+1}\backslash \mathcal{S} $, the function $T(\cdot,t)$ with
\begin{equation}
\label{T}
T(x,t) = \int_{0}^t e^{-\Lambda(t-s)} {\bf 1}_{n+1} y_x(s)ds,
\end{equation}
is {\em injective} for any $t\in \cali$ with $t> t_u$, in which $\Lambda := \diag\{\lambda_1, \ldots, \lambda_{n+1}\}$ and ${\bf 1}_{n+1}$ is an $(n+1)$-dimensional vector of ones.

\qed
\end{lemma}

The result presented in \cite[Theorem 3]{BERAND} pertains to systems with {\em inputs}. But, as explained in Footnote 2 of \cite{BERAND} it applies as well to nonautonomous systems of the form \eqref{nlsys}. Besides, the result in \cite[Theorem 3]{BERAND} requires the backward-distinguishable condition \emph{for all} $t\ge t_u$ rather than a bounded interval $\cali$, and thus provides the stronger conclusion of the injectivity of $T(\cdot,t)$ also for all $t\in [t_u,+\infty)$. Even though this modified statement is not a priori necessarily equivalent to Lemma \ref{lem1}, the above ``local'' version can be proved in the same way as \cite[Theorem 3]{BERAND} without additional arguments needed.

\section*{Acknowledgments}
The authors would like to express their gratitude to the three anonymous reviewers for their careful reading of our paper and for pointing out two important corrections to it. We are particularly grateful to Reviewer 1 who did an extraordinary job.

\end{document}